\documentclass[11pt,letterpaper]{article}
\usepackage[margin=1in]{geometry}
 \usepackage{authblk}
\usepackage{graphicx,amssymb,amsmath,xcolor}
\usepackage{subcaption}  

\newtheorem{theorem}{Theorem}[section]

\newtheorem{lemma}[theorem]{Lemma}

\newenvironment{proof}{\textbf{Proof:}}{\hfill$\square$}

\newcommand{\polylog}{\mathrm{polylog}}

\title{The Maximum Clique Problem in a Disk Graph Made Easy} 

\author[1]{J. Mark Keil}
\author[1]{Debajyoti Mondal}
\affil[1]{Department of Computer Science, University of Saskatchewan, Saskatoon,  Canada} 
 
\date{}

\begin{document}

\maketitle

\begin{abstract}
A disk graph is an intersection graph of disks in $\mathbb{R}^2$. Determining the computational complexity of finding a maximum clique in a disk graph is a long-standing open problem.  In 1990, Clark, Colbourn, and Johnson gave a polynomial-time algorithm for computing a maximum clique in a unit disk graph. However, finding a maximum clique when disks are of arbitrary size is widely believed to be a challenging open problem. The problem is open even if we restrict the disks to have at most two different sizes of radii, or restrict the radii to be within $[1,1+\varepsilon]$ for some $\epsilon>0$. 
In this paper, we provide a new perspective to examine adjacencies in a disk graph that helps obtain the following results. 

\begin{enumerate}
    \item[-] We design an $O(n^{2k} poly(n))$-time algorithm to find a maximum clique in a $n$-vertex disk graph with $k$ different sizes of radii. This is polynomial for every fixed $k$, and thus settles the open question for the case when $k=2$.
    \item[-]  Given a set of $n$ unit disks, we show how to compute a maximum clique inside each possible axis-aligned rectangle determined by the disk centers in $O(n^5\log n)$-time. This is at least a factor of $n^{4/3}$ faster than applying the fastest known algorithm for finding a maximum clique in a unit disk graph for each rectangle independently.  
    \item[-] We give an $O(n^{2rk} poly(n,r))$-time algorithm to find a maximum clique in a $n$-vertex ball graph with $k$ different sizes of radii where the ball centers lie on $r$ parallel planes. This is polynomial for every fixed $k$ and $r$, and thus contrasts the previously known NP-hardness result for finding a maximum clique in an arbitrary ball graph.
\end{enumerate}


\end{abstract}

\section{Introduction}
 \label{sec:intro}
A \emph{geometric intersection graph} consists of a set of geometric objects as vertices and a set of edges that are determined by the intersection of these objects, i.e., two vertices are considered to be adjacent if and only if the corresponding objects intersect. Geometric objects of different shapes and their intersection graphs are often used to model applied contexts, e.g., unit disks to model problems in wireless networks~\cite{DBLP:conf/waoa/Fishkin03,huson1995broadcast}, line segments (time intervals) to model task scheduling~\cite{Kleinberg+Tardos:06a}, trapezoids to model channel routing in VLSI design~\cite{DBLP:journals/dam/FelsnerMW97}, and so on. Many NP-hard graph problems are known to admit polynomial-time solutions on various types of geometric intersection graphs (see e.g.,~\cite{DBLP:journals/comgeo/BoseCKM0MS22,DBLP:journals/dm/ClarkCJ90}). There also exist cases where a graph problem is known to be NP-hard, but its time complexity status is open for some intersection graph class (see e.g. the open problems posed in~\cite{DBLP:journals/dcg/AdhikaryBMR23,DBLP:journals/jacm/BonamyBBCGKRST21,DBLP:journals/corr/abs-2107-05198}). In this paper, we restrict our attention to one such scenario, where the problem of interest is finding a \emph{maximum clique}, i.e., a maximum subset of pairwise adjacent vertices, and the intersection graph we examine is a \emph{disk graph}. Here, a \emph{disk graph} $G$ is an intersection graph of disks in $\mathbb{R}^2$, where each vertex of $G$ corresponds to a disk and two vertices are adjacent in $G$ if and only if their disks intersect. 

The maximum clique problem is NP-hard for many well-known intersection graph classes such as intersection graph  of rays~\cite{DBLP:journals/dcg/CabelloCL13}, grounded strings~\cite{DBLP:journals/corr/abs-2107-05198}, triangles~\cite{ambuhl2005clique}, ellipses~\cite{ambuhl2005clique}, and a combination of axis-aligned rectangles and unit disks~\cite{bonnet_et_al:LIPIcs:2020:13258,DBLP:conf/compgeom/EspenantKM23}, whereas polynomial-time solvable for circle graphs~\cite{Nash10}, trapezoid graphs~\cite{DBLP:journals/dam/FelsnerMW97},
circle trapezoid graphs~\cite{DBLP:journals/dam/FelsnerMW97}, unit disk graphs~\cite{DBLP:journals/dm/ClarkCJ90},  axis-aligned rectangle intersection graphs~\cite{DBLP:journals/jal/ImaiA83}, and so on. However, determining the time complexity of computing a maximum clique in a disk graph is a long-standing open question~\cite{ambuhl2005clique,BRS06,DBLP:conf/waoa/Fishkin03}. In a seminal paper published in 1990, Clark, Colbourn, and Johnson~\cite{DBLP:journals/dm/ClarkCJ90}  gave a polynomial-time algorithm for the case of unit disk graphs. However, the time-complexity question for general disk graphs has remained open since then. Although this was not posed as an open problem in~\cite{DBLP:journals/dm/ClarkCJ90}, it has long been known to be a challenging problem~\cite{BRS06}, even before it was explicitly posed as an open problem  (e.g., by Fishkin~\cite{DBLP:conf/waoa/Fishkin03}, Amb\"{u}hl and Wagner~\cite{ambuhl2005clique} and Cabello~\cite{c1,c2}). 
 In fact, the problem 
 has been emphasized in various ways in the literature, e.g., as ``an intriguing open question'', 
 as ``a notorious open question in computational geometry''~\cite{DBLP:journals/jacm/BonamyBBCGKRST21}, as being ``elusive with no new positive or negative results''~\cite{BRS06},  and several times as a ``long-standing open problem''~\cite{DBLP:conf/compgeom/EspenantKM23,DBLP:journals/jco/Grelier22,DBLP:phd/basesearch/Grelier22}.

Although no polynomial-time exact solution is known for finding a maximum clique in general disk graphs, a 2-approximate solution~\cite{ambuhl2005clique} can be obtained by computing a stabbing set of four points. Such a stabbing set, i.e., four points in the plane that hit all the given disks, is known to exist~\cite{danzer1986losung,stacho1981gallai} and can be computed in linear time~\cite{DBLP:journals/dcg/CarmiKM23}. However, obtaining a better approximation ratio appears challenging. Cabello~\cite{c1,c2} even asked whether a 1.99 approximation can be achieved if we restrict the input to disks that have at most two different sizes of radii.

Several recent research has shown interesting progress on various fronts of the problem. Clark et al.'s $O(n^{4.5})$-time algorithm for finding a maximum clique in a unit disk graph~\cite{DBLP:journals/dm/ClarkCJ90}, which was improved first to $O(n^{3.5}\log n)$-time~\cite{breu} and then to $O(n^{3}\log n)$-time~\cite{DBLP:conf/wg/Eppstein09}, has been improved further in SoCG'23 to $O(n^{2.5}\log n)$~\cite{DBLP:conf/compgeom/EspenantKM23}. Chan~\cite{chan} observed that the running time of~\cite{DBLP:conf/compgeom/EspenantKM23} can be expressed as $O(n \cdot h(m)^{1+o(1)})$, where $h(m)$ is the time to compute a maximum matching in a bipartite graph that has a biclique cover\footnote{A \emph{biclique cover} of a bipartite graph $G$ is a set of complete bipartite subgraphs of $G$ that cover the edges of $G$. The size of the cover is the number of elements in the set.} of size $m$, and  $m$ in this context is known to be $O(n^{4/3} \polylog(n))$~\cite{DBLP:journals/siamcomp/KatzS97}. One can leverage such a biclique cover to compute a maximum matching in $O(m^{1+o(1)})$ time~\cite{DBLP:journals/jcss/FederM95}. Therefore, the overall running time becomes $O(n^{7/3+o(1)})$. In FOCS'18, Bonamy et al.~\cite{DBLP:conf/focs/BonamyBBCT18} gave a QPTAS, a randomized EPTAS (and thus a fixed-parameter-tractable algorithm~\cite{bonnet_et_al:LIPIcs:2020:13258}), and a subexponential algorithm for computing a maximum clique in a disk graph. They also showed how to extend these results for \emph{unit ball graphs}, which are intersection graphs of 3-dimensional balls of unit radius. In SODA'22, Lokshtanov et al.~\cite{DBLP:conf/soda/LokshtanovPSXZ22} designed further subexponential-time fixed-parameter-tractable algorithms for many other fundamental graph problems on disk graphs. While the maximum clique problem is open for disk graphs with arbitrary radii, it is NP-hard for ball graphs~\cite{DBLP:conf/focs/BonamyBBCT18}. The hardness result holds in a very restricted setting with all radii falling within the interval  $[1,1+\epsilon]$,  $\epsilon>0$, where even a subexponential-time approximation scheme is unlikely unless the exponential-time-hypothesis conjecture~\cite{DBLP:conf/focs/ImpagliazzoPZ98} fails.

\subsection*{Our Contribution}

We show that a maximum clique in a unit disk graph can be obtained by examining \emph{slabs} (i.e., regions that are bounded by two parallel lines), whereas all prior algorithms depend on a more constrained lens-shaped region. We show how our new slab-based idea can be extended to compute a maximum clique in polynomial time when the number of different radii sizes is fixed, and even to address the case of ball graphs in some restricted settings.  Specifically, we obtain the following results, where in all cases, we assume that an arrangement of disks, i.e., a disk representation of the disk graph, is given as an input.

\smallskip 
\textbf{Disk graph with $k$ Different Sizes of Radii}:
We give an $O(n^{2k} (f(n)+n^2))$-time algorithm to find a maximum clique in a disk graph where $k$ is the number of different types of radii and $f(n)$ is the time to compute a maximum matching in a $n$-vertex bipartite graph. Since $f(n)$ is polynomial in $n$~\cite{DBLP:journals/siamcomp/HopcroftK73}, for every fixed $k$, the running becomes polynomial. This settles the open question posed by Cabello~\cite{c1,c2} on whether a polynomial-time algorithm exists when $k=2$.
 
\smallskip 
\textbf{Range Query in a Unit disk graph}: We show that our slab-based idea can help the range query version of the maximum clique problem, where the input unit disks should be preprocessed such that several queries that may appear at a later time can be answered efficiently. We examine the case of \emph{axis-aligned rectangular query}, where a maximum clique needs to be reported over disks with centers in the query rectangle. 
A natural approach for handling rectangular range queries is to 
precompute the maximum cliques for all possible axis-aligned rectangles (determined by the given disk centers)  such that given a query $R$, a solution can be obtained first by identifying the smallest enclosing rectangle $R'$ for the disk centers in $R$, and then by a table look-up using $R'$. Applying an existing lens-based algorithm on all $O(n^4)$ possible axis-aligned rectangles determined by the input disk centers takes $O(n^{19/3+o(1)})$ time.  In contrast, we can use our slab-based idea to precompute all such solutions in $O(n^5\log n)$, achieving a speed up by at least a factor of $n^{4/3}$.

\smallskip 
\textbf{Ball graphs with $k$ Different Sizes of Radii}: Given a set of $n$ balls in the Euclidean plane with $k$ different sizes of radii where the ball centers lie on $r$ parallel planes, we show how to compute a maximum clique in the corresponding ball graph in $O(n^{2rk} (f(n)+n^2r))$ time. Here $f(n)$ is the time to compute a maximum matching in a $n$-vertex bipartite graph. For fixed $k$ and $r$, the running time becomes polynomial in $n$. We then show that the restriction of the planes being parallel can be removed as long as the input planes are all perpendicular to a different plane. This result contrasts the known NP-hardness result~\cite{DBLP:journals/jacm/BonamyBBCGKRST21} for finding a maximum clique in a ball graph where the radii of the balls are very close to each other, i.e., in the interval  $[1,1+\varepsilon]$, where $\varepsilon>0$.

\section{Preliminaries}
In this section we discuss some notation and preliminary results.

By $D_{a,r}$, we denote a disk with a center at point $a$ and radius $r$. Given a pair of disks $D_{a,r}$ and $D_{b,r'}$, we refer to their common intersection region as a \emph{lens} (Figure~\ref{fig:u}(a)). 
Let $G$ be a disk graph. For simplicity, sometimes we say a pair of disks are adjacent in $G$, which means that the disks intersect and their corresponding vertices are adjacent in $G$.  By $B_{a,r}$, we denote a ball with a center at point $a$ and radius $r$.
 
\begin{figure}[h]
    \centering
    \includegraphics[width=\textwidth]{"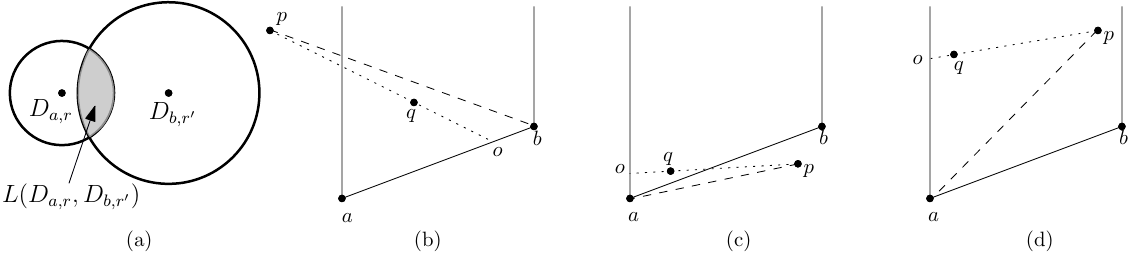"}
    \caption{ Illustration for (a) a lens, (b)--(d)  Illustration for Lemma~\ref{lem:u}. }
    \label{fig:u}
\end{figure}

For a pair of points $a$ and $b$, we 
denote the line passing through them as $\ell_{ab}$. By $ab$ and $|ab|$ we denote the 
line segment with endpoints $a,b$ and the 
Euclidean distance between $a$ and $b$, respectively. By $\ell^h_{p}$ and $\ell^v_{p}$ we denote the horizontal and vertical lines through a point $p$, respectively. By an \emph{upper slab} $U_{ab}$ of $ab$, we denote the region above $ab$, which is bounded by the lines $\ell^v_{a}$ and $\ell^v_{b}$. Similarly, by a  \emph{lower slab} $\overline{U}_{ab}$ we denote the region below $ab$, which is bounded by $\ell^v_{a}$ and $\ell^v_{b}$.


\begin{lemma}\label{lem:u}
Let $ab$ be a line segment and let $U_{ab}$ be the upper slab of $ab$. Let  $q$ be a point in $U_{ab}$ and let $p$ be a point (not necessarily in $U_{ab}$) with a $y$-coordinate equal to or larger than the $y$-coordinate of $q$. Then $|pq|\le
\max\{|pa|,|pb|\}$.
\end{lemma}
\begin{proof}
We distinguish two scenarios depending on whether $ab$ is vertical or not. 

\textit{Scenario 1 (The $x$-coordinates of $a$ and $b$ are different):} Without loss of generality assume that  $a$ has a strictly smaller $x$-coordinate than that of $b$. Since the $y$ coordinate of $p$ is equal to or larger than the $y$-coordinate of $q$, the ray that starts at $p$ and passes through $q$ must hit the boundary of $U_{ab}$. Let $o$ be the point of intersection when the ray exits $U_{ab}$ (e.g., Figure~\ref{fig:u}(b)--(d)). It suffices to show that $|po|\le \max\{|pa|,|pb|\}$. 

Consider first the case when $o$ lies on $ab$ (e.g., Figure~\ref{fig:u}(b)). Let $m$ be the point on line $\ell_{ab}$ that minimizes the distance $|pm|$, i.e., $\ell_{pm}$ is perpendicular to $\ell_{ab}$. We now move $o$ away from $m$ along line $\ell_{ab}$ until it hits either $a$ or $b$. Since we are moving $o$ away from $m$, the length $|po|$ increases monotonically, and the largest length it can attain is $\max\{|pa|,|pb|\}$. 
 
Consider now the case when $o$ lies on the vertical sides of $U_{ab}$ (e.g., Figure~\ref{fig:u}(c)--(d)). We move $o$ downward along the side of $U_{ab}$ until it hits either $a$ or $b$. Since $p$ has a higher $y$-coordinate than $q$, moving $o$ downward increases the length $|po|$ monotonically, and the largest length it can attain is $\max\{|pa|,|pb|\}$.

\textit{Scenario 2 (The $x$-coordinates of $a$ and $b$ are the same):}
In this case, $ab$ is a vertical segment. If $a$ has a smaller $y$-coordinate than that of $b$, then $|pq|\le |pa|$. Otherwise, $|pq|\le |pb|$.
\end{proof}

\section{Maximum Clique in Disk Graphs with $k$ Different  Radii}
\label{sec:kradii}
In this section we give an $O(n^{2k} (f(n)+n^2))$-time algorithm to find a maximum clique in a disk graph where $k$ is the number of different types of radii and $f(n)$ is the time to compute a maximum matching in a $n$-vertex bipartite graph. 

We first introduce some notation. Let $\mathcal{D}_k$ be a disk graph where the number of different types of radii is at most $k$. For convenience, we denote these different types of radii as $r_1,r_2,\ldots, r_k$, where for every $1\le i<j\le k$, we have $r_i< r_j$. By a 
\emph{type-$i$ disk} we denote a disk of radius $r_i$. By $C_i$ we denote a clique that contains only disks of type $i$. Let $\mathcal{C}$ be a maximum clique of $\mathcal{D}_k$. By $\mathcal{C}_i$ we denote a maximal clique in $\mathcal{C}$ where all disks are of type $i$. 

The idea of the algorithm is as follows. We guess the number of different types of radii that may appear in a maximum clique and we take the maximum over all the solutions computed from these $2^k$ guesses. 
We now describe how a solution is computed for a particular guess for a set of radii that may appear in the maximum clique $\mathcal{C}$. For each disk type $i$, we guess two disk centers $a_i,b_i$ from $\mathcal{C}_i$, where $a_i$ is the leftmost (i.e., with the smallest $x$-coordinate) and $b_i$ is the rightmost (i.e., with the largest $x$-coordinate) over all the disk centers in $\mathcal{C}_i$. {\color{black}Note that  $a_i$ may coincide with $b_i$ if $\mathcal{C}_i$ contains only one disk, or if the centers of all disks in $\mathcal{C}_i$ are on a vertical line}. Let $\Psi$ be the set of these $2k$ disks.  For every upper slab $U_{a_ib_i}$, we construct a set $X_i$ by taking every disk that has its center in $U_{a_ib_i}$ and intersects all the disks of $\Psi$.  We show that the union of these disks, i.e., $X= (X_1\cup X_2\cup \ldots\cup X_k)$, is a clique in $\mathcal{D}_k$. Similarly, for each lower slab $\overline{U}_{a_ib_i}$, we construct a set $Y_i$ by taking every disk that has its center in $\overline{U}_{a_ib_i}$ and intersects all the disks of $\Psi$. We show that their union $Y=(Y_1\cup Y_2\cup \ldots\cup Y_k)$ is a clique in $\mathcal{D}_k$. Since the disks of $\Psi$ are assumed to be in $\mathcal{C}$, it is straightforward to observe that $\mathcal{C}$ is a subset of the disks in $(\Psi\cup X\cup Y)$. Since the complement of the disk graph determined by the disks $(X \cup Y)$ is a bipartite graph $H$,  we can compute the maximum clique $\mathcal{C}$ by from a maximum bipartite matching in $H$. 

It now suffices to show that the disks in $X$ (similarly, the disks in $Y$) are mutually adjacent. 

\begin{lemma}\label{lem:pairmagic}
    For every $i,j$, where $1\le i,j\le k$, the disks in $X_i\cup X_j$ are mutually adjacent. 
\end{lemma}
\begin{proof}
Let  $D_{p,r_i}$ and $D_{q,r_j}$ be two disks, where $p$ belongs to $X_i$ and $q$ belongs to $X_j$. We now show that $D_{p,r_i}$ and $D_{q,r_j}$ must mutually intersect.

Consider first the case when $i=j$. Without loss of generality assume that the $y$-coordinate of $p$ is larger than or equal to that of $q$. Then by Lemma~\ref{lem:u}, $|pq|\le \max\{|pa_i|,|pb_i|\}$. By the construction of $\Psi$, the disk $D_{p,r_i}$ intersects $D_{a_i,r_i}$ and $D_{b_i,r_i}$. We thus have  $|pq|\le \max\{|pa_i|,|pb_i|\} \le 2r_i$, and since $p$ and $q$ correspond to type-$i$ disks, they must intersect.  

Consider now the case when $i\not=j$. 

If the $y$-coordinate of $p$ is larger than or equal to that of $q$, then we apply Lemma~\ref{lem:u}  to obtain $|pq|\le \max\{|pa_j|,|pb_j|\}$.  By the construction of $\Psi$, the disk $D_{p,r_i}$ intersects $D_{a_j,r_j}$ and $D_{b_j,r_j}$. We thus have  $|pq|\le \max\{|pa_j|,|pb_j|\} \le r_i+r_j$. Consequently, $D_{p,r_i}$ and $D_{q,r_j}$ mutually intersect. 
 
If the $y$-coordinate of $p$ is smaller than that of $q$, then we apply Lemma~\ref{lem:u} by swapping the role of $p$ and $q$, i.e., by using the condition that $p\in U_{a_i,b_i}$ whereas $q$ has a larger $y$-coordinate than that of $p$. We thus have $|pq|\le \max\{|qa_i|,|qb_i|\}$. By the construction of $\Psi$, we have  $\max\{|qa_i|,|qb_i|\} \le r_i+r_j$. Consequently, $D_{p,r_i}$ and $D_{q,r_j}$ mutually intersect. 
\end{proof}

We now consider the running time. There are at most ${n\choose 0}+{n\choose 1}+{n\choose 2}=O(n^2)$ ways to guess the pair of disks for a particular disk type, irrespective of whether these two disks are distinct or not. Therefore, the number of ways we can guess $k$  pairs is 
$O(n^{2k})$. For each guess of $k$  pairs, it is straightforward to construct the sets $X$ and $Y$ in $O(kn)\in O(n^2)$ time, and the disk graph determined by the disks $(X \cup Y)$ in $O(n^2)$ time. Let $f(n)$ be the time for computing a maximum matching in a $n$-vertex bipartite graph.  Then the   running time is $O(n^{2k} (f(n)+n^2))$. 

The following theorem summarizes the result of this section.

\begin{theorem}
Given a set of $n$ disks in the Euclidean plane with $k$  different types of radii, a maximum clique in the corresponding disk graph can be computed 
   in $O(n^{2k} (f(n)+n^2)) $ time. Here $f(n)$ is the time to compute a maximum matching in a $n$-vertex bipartite graph. 
\end{theorem}



\section{Rectangular Range Query over Unit Disks}
In this section, we consider the problem of reporting a maximum clique given a rectangular range query over an arrangement of unit disks. To this end, we show how to compute a maximum clique for every axis-aligned rectangle (determined by the input disk centers) in $O(n^5\log n)$ time, which is at least a factor of $n^{4/3}$ faster than applying the algorithm of~\cite{DBLP:conf/compgeom/EspenantKM23} separately for each axis-aligned rectangle. 

To better explain such advantages of our technique, we first give an overview of the existing algorithms 
to compute a maximum clique in a unit disk graph. The idea of Clark et al.'s $O(n^{4.5})$-time algorithm~\cite{DBLP:journals/dm/ClarkCJ90} is to guess the farthest pair of vertices $(v,w)$ in the clique, and then 
to find a maximum clique $C$ that contains both $v$ and $w$. Since $v$ and $w$ are the farthest pair, the set $S$ of vertices that are adjacent to both $v,w$ must lie in a lens-shaped region, i.e., the common intersection region of the two disks centered at $v$ and $w$ with radius $|vw|$. Clark et al. showed that the lens can be partitioned by the line segment $vw$ into two halves where the disks with centers in the same half are pairwise intersecting. Consequently, the complement of the disk graph corresponding to $S$ is a bipartite graph $H$, and the maximum clique $C$ can be computed by finding a maximum independent set in $H$. The running time for the unit disk case has subsequently been improved to $O(n^{3.5}\log n)$ by using non-trivial data structures~\cite{breu}, then to $O(n^{3}\log n)$ by examining the lenses in a particular order so that the solution for each lens does not need to be computed from scratch~\cite{DBLP:conf/wg/Eppstein09,DBLP:journals/dcg/EppsteinE94}, and finally, to $O(n \cdot h(m)^{1+o(1)})$ by a combination of divide-and-conquer and plane sweep approach using a circular arc~\cite{DBLP:conf/compgeom/EspenantKM23,chan}. 
Here $h(m)$ is the time to compute a maximum matching in $H$, i.e., the complement of the disk graph inside a lens, and $m$ is the size of the biclique cover of $H$. Since $m\in O(n^{4/3} \polylog(n))$~\cite{DBLP:journals/siamcomp/KatzS97} and a maximum matching can be computed using a maximum-flow algorithm in $O(m^{1+o(1)})$ time~\cite{DBLP:journals/jcss/FederM95}, the running time becomes $O(n^{7/3+o(1)})$. Consider now the scenario of computing a maximum clique for every possible rectangle determined by the input disk centers. Since lenses examined by these algorithms are not necessarily axis-aligned, finding an order of the lenses to compute solutions for the rectangles by dynamic point insertion and deletion appears to be challenging. A straightforward approach that solves each rectangular region independently takes $O(n^{6.33+o(1)})$ time. 


The idea of our algorithm is to maintain solutions for all possible slabs where one can insert and delete points as necessary to find the solutions for all possible rectangles within the slab. For convenience, we assume that the centers of the disks are in general position, i.e.,  no two centers have the same $y$-coordinates. We now describe the details.

Consider first an alternative (slower) approach that uses slabs to compute a maximum clique that guesses the leftmost and rightmost points $a,b$ of a maximum clique. For every such guess, one can find a maximum clique by considering the set $S$ of disks that are adjacent to both $a,b$ and have their centers in the slab bounded by $\ell^v_a$ and $\ell^v_b$. By Lemma~\ref{lem:u}, the disks of $S$ that have their centers in $U_{ab}$ (similarly, in $\overline{U}_{a,b}$) are mutually adjacent. Consequently, the complement of the disk graph underlying $S$ is a bipartite graph $H$, and a maximum clique can be computed by finding a maximum independent set in $H$. 


We now consider maintaining the solutions for all possible slabs under point insertion and deletion. Initially, the solution for each slab is set to null. We now sweep the plane upward with two horizontal lines $\ell^b$ and $\ell^u$ in two phases. 

In the first phase, $\ell^b$ is placed so that it passes through the disk center with the lowest $y$-coordinate, and $\ell^u$ is used to sweep the plane upward starting at $\ell_b$.  Each time $\ell^u$ moves to a new point, a point is inserted into the corresponding slabs and their solutions are updated. Every time the horizontal slab between  $\ell^b$ and $\ell^u$ contains the guessed points of a slab (i.e., the two points determining a slab), we obtain a solution for a new rectangular region, and the corresponding  solution 
is stored in a table. It is straightforward to observe that each insertion changes $O(n^2)$ slabs, and hence we have $O(n^3)$ insertion operations in total.

In the next phase, we sweep the plane upward using  $\ell^b$. Each time $\ell^b$ moves to a new point, a point is deleted from the corresponding rectangular regions we considered in the first phase. A deletion changes $O(n^3)$ previous solutions, i.e., the solutions corresponding to $O(n^2)$ slabs and $O(n)$ regions in each slab determined by the positions of $\ell_b$. Therefore, we have $O(n^4)$ deletion operations in total.

Each insertion and deletion operation updates existing lenses, each of which is determined by a pair of guessed points. It is known that such updates can be done in $O(n\log n)$ time by performing an alternating path search to update the maximum independent set in the bipartite graph determined by the complement of the disk graph inside the lens~\cite{DBLP:journals/dcg/EppsteinE94,DBLP:conf/wg/Eppstein09}. For $O(n^4)$ update operations, the running time becomes $O(n^5\log n)$. By $S[p_\ell,p_r,p_t,p_b]$  we denote the size of a maximum clique inside a rectangle $R$ with its left, right, top and bottom sides being determined by the lines through the disk centers $p_\ell,p_r,p_t,p_b$, respectively, where $|p_\ell,p_r|$ is at most two units and $p_\ell,p_r$ appear on the left and right sides of $R$. If $|p_\ell,p_r|$ is larger than two units, then $S[p_\ell,p_r,p_t,p_b]$ is 0.

We now compute a maximum clique for each axis-aigned rectangle determined by the input disk centers. We first compute two sorted orders of the disk centers, one based on increasing $x$-coordinates and the other based on increasing $y$-coordinates.  We now use a dynamic programming approach, where the base cases are the $O(n^4)$ rectangles that do not contain any disk center in their proper interior. The solution for the base cases can be computed and stored in a table in $O(n^4)$ time. Let $M[p_\ell,p_r,p_t,p_b]$ be the maximum clique of a rectangle $R$ with its left, right, top and bottom sides being determined by the lines through $p_\ell,p_r,p_t,p_b$, respectively. Note that the disk centers $p_\ell,p_r,p_t,p_b$ do not necessarily appear on the sides of $R$. Let $p'_\ell$ and $p'_r$ be the disk centers immediately after and before $p_\ell$ and $p_r$, respectively, which can be determined from the precomputed sorted order. Similarly, $p'_t$ and $p'_b$ are the disk centers immediately before and after $p_t$ and $p_b$, respectively. We now can find a solution using $O(1)$ table look-ups as follows.

\begin{align*}
M[p_\ell,p_r,p_t,p_b]=
\begin{cases}
\max\{S[p_\ell,p_r,p_t,p_b], M[p'_\ell,p_r,p_t,p_b], M[p_\ell,p'_r,p_t,p_b],\\ 
\hspace{1cm}M[p_\ell,p_r,p'_t,p_b], M[p_\ell,p_r,p_t,p'_b]\}, \text{if $p_\ell$ and $p_r$ both appear on $R$}   \\
\max\{M[p'_\ell,p_r,p_t,p_b], M[p_\ell,p'_r,p_t,p_b],M[p_\ell,p_r,p'_t,p_b], M[p_\ell,p_r,p_t,p'_b]\}, \text{otherwise.}   \\
\end{cases}
\end{align*}

Since there are $O(n^4)$ cells in the table and each entry is computed by $O(1)$ table look-ups which are determined in $O(\log n)$ time, the overall running time remains $O(n^5\log n)$.

\begin{theorem}
Given a set of $n$ unit disks in the Euclidean plane, a maximum clique for every possible axis-aligned rectangle determined by the given disk centers can be computed in $O(n^5\log n)$ time.
\end{theorem}
 
\section{Maximum Clique in Ball Graphs with $k$ Different Radii}

We now give an  $O(n^{2rk} (f(n)+n^2r))$-time algorithm for computing a maximum clique in a ball graph, where the centers of the balls are contained in $r$ planes which are all perpendicular to a different plane (Section~\ref{sec:para2}). For simplicity, we first consider a scenario when the input planes are parallel to each other (Section~\ref{sec:para}).

\subsection{Ball Centers are on $r$ Parallel Planes}\label{sec:para}

Without loss of generality, we assume that the $r$ input  planes are parallel to the xy-plane. We first show that a version of Lemma~\ref{lem:u} holds in three dimensions, as follows. 

\begin{lemma}\label{lem:u3d}
Let $ab$ be a line segment on a plane $L$, where $L$ is parallel to the xy-plane, and let $U_{ab}$ be the upper slab of $ab$ on $L$. Let $q$ be a point in $U_{ab}$. Let $p$ be a point (not necessarily on $L$) with a $y$-coordinate equal to or larger than the $y$-coordinate of $q$. Then $|pq|\le
\max\{|pa|,|pb|\}$. 
\end{lemma}
\begin{proof} 
Let $p'$ be the projection of $p$ on $L$ (Figure~\ref{fig:poly}(a)). Then the $y$-coordinate of $p'$ is equal to or larger than the $y$-coordinate of $q$. 
By Lemma~\ref{lem:u},   $|p'q|\le \max\{|p'a|,|p'b|\}$ and thus 
\begin{align*}
&\sqrt{|pq|^2-|pp'|^2} \le \max\{\sqrt{|pa|^2-|pp'|^2},\sqrt{|pb|^2-|pp'|^2}\}\\
\implies &{|pq|^2-|pp'|^2} \le \max \{{|pa|^2-|pp'|^2},{|pb|^2-|pp'|^2}\}\\
\implies &{|pq|} \le  \max\{|pa|,|pb|\}.
\end{align*}  
\end{proof}

Let $B_k$ be a ball graph with $k$ different types of radii. 
Similar to Section~\ref{sec:kradii}, we guess the number of different types of radii that may appear in a maximum clique and we take the maximum over all the solutions computed from these $2^k$ guesses. For each ball type $i$, we guess at most $2r$ ball centers $\{a^1_i,b^1_i\},\ldots, \{a^r_i,b^r_i\}$ from $\mathcal{C}_i$. Here $a^j_i$ is the leftmost and $b^j_i$ is the rightmost over all the ball centers in $\mathcal{C}_i$ on the $j$th plane, where $1\le j\le r$. {\color{black}Note that  $a^j_i$ may sometimes coincide with $b^j_i$}. Let $\Psi$ be the set of these $2rk$ balls.  For every upper slab $U_{a^j_ib^j_i}$ on the $j$th plane, we construct a set $X^j_i$ by taking every ball that has its center in $U_{a^j_ib^j_i}$ and intersects all the balls of $\Psi$.  We show that the union of these balls, i.e., $X= \bigcup\limits_{1\le i\le k}(X^1_i\cup X^2_i\cup\ldots\cup X^r_i)$, is a clique in $\mathcal{D}_k$. Similarly, for each lower slab $\overline{U}_{a_ib_i}$, we construct a set $Y^j_i$ by taking every ball that has its center in $\overline{U}_{a^j_ib^j_i}$ and intersects all the balls of $\Psi$. We show that their union $Y= \bigcup\limits_{1\le i\le k}(Y^1_i\cup Y^2_i\cup\ldots\cup Y^r_i)$ is a clique in $\mathcal{D}_k$. 
It now suffices to show that the balls in $X$ (similarly, the balls in $Y$) are mutually adjacent.

\begin{figure}[h]
    \centering
    \includegraphics[width=.8\textwidth]{"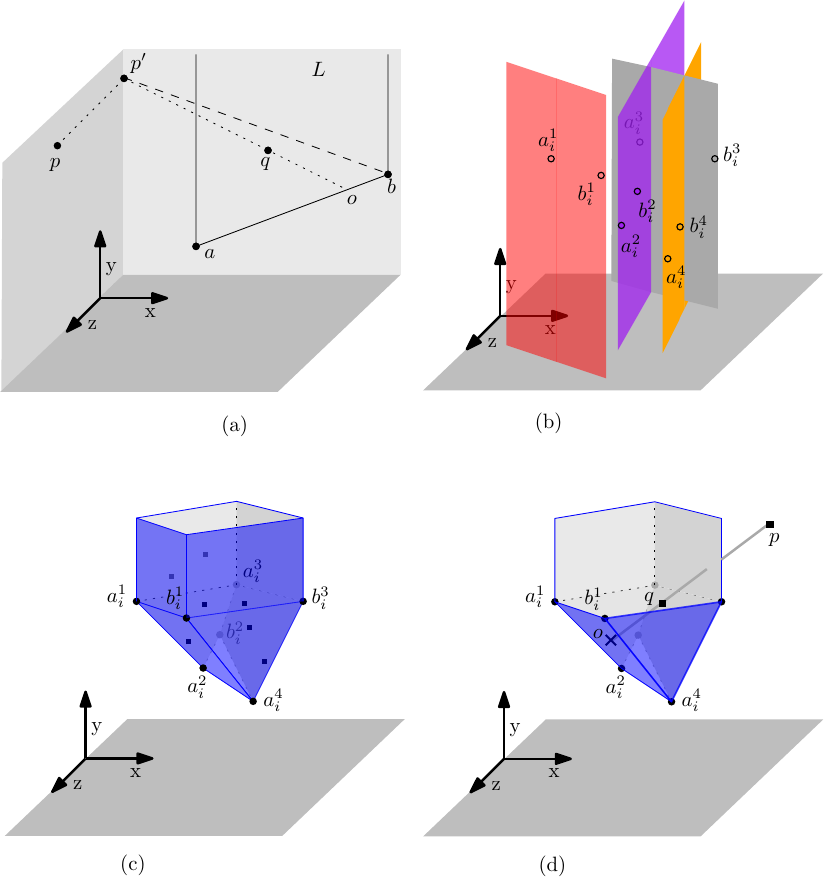"}
    \caption{ Illustration for (a) Lemma~\ref{lem:u3d}, and (b)--(d) Lemma~\ref{lem:polymagic}. }
    \label{fig:poly}
\end{figure}

\begin{lemma}\label{lem:pairmagic3d}
    For every $i,j$, where $1\le i,j\le k$, the balls in $X^g_i\cup X^h_j$, where $1\le g\le h\le r$, are mutually adjacent. 
\end{lemma}
\begin{proof}
Let $B_{p,r_i}$ and $B_{q,r_j}$ be two balls, where $p$ belongs to $X^g_i$ and $q$ belongs to $X^h_j$. We now show that $B_{p,r_i}$ and $B_{q,r_j}$ must mutually intersect. If $g=h$, then the proof follows from Lemma~\ref{lem:pairmagic}. Therefore, we may assume that $g\not =h$. However, the proof in this case is the same as that of Lemma~\ref{lem:pairmagic} except that we use Lemma~\ref{lem:u3d} instead of Lemma~\ref{lem:u} for the arguments, as follows. 

Consider first the case when $i=j$. Without loss of generality assume that the $y$-coordinate of $p$ is larger than or equal to that of $q$. Then by Lemma~\ref{lem:u3d}, $|pq|\le \max\{|pa^h_j|,|pb^h_j|\}$. By the construction of $\Psi$, the ball $B_{p,r_i}$ intersects $B_{a^h_j,r_i}$ and $B_{b^h_j,r_i}$. We thus have  $|pq|\le \max\{|pa^h_j|,|pb^h_j|\} \le 2r_j$, and since $p$ and $q$ correspond to type-$j$ balls, they must intersect.  

Consider now the case when $i\not=j$. If the $y$-coordinate of $p$ is larger than or equal to that of $q$, then we apply Lemma~\ref{lem:u3d} 
to obtain $|pq|\le \max\{|pa^h_j|,|pb^h_j|\}$.  By the construction of $\Psi$, the ball $B_{p,r_i}$ intersects $B_{a^h_j,r_j}$ and $B_{b^h_j,r_j}$. We thus have  $|pq|\le \max\{|pa^h_j|,|pb^h_j|\} \le r_i+r_j$. Consequently, $B_{p,r_i}$ and $B_{q,r_j}$ mutually intersect. If the $y$-coordinate of $p$ is smaller than that of $q$, then we apply Lemma~\ref{lem:u3d} by swapping the role of $p$ and $q$, i.e., by using the condition that $p\in U_{a^g_i,b^g_i}$ whereas $q$ has a larger $y$-coordinate than that of $p$. We thus have $|pq|\le \max\{|qa^g_i|,|qb^g_i|\}$. By the construction of $\Psi$, we have  $\max\{|qa^g_i|,|qb^g_i|\} \le r_i+r_j$. Consequently, $B_{p,r_i}$ and $B_{q,r_j}$ mutually intersect. 
\end{proof}

We now consider the running time. For a single plane, there are at most $O(n^2)$ ways to guess the pair of balls for a particular ball type, and thus $O(n^{2k})$ ways for $k$  pairs considering all types. Since we have $r$ planes, the total number of guesses is $O(n^{2rk})$. For each guess of $rk$ pairs, it is straightforward to construct the sets $X$ and $Y$ in $O(rkn)\in O(n^2r)$ time, and the ball graph determined by the balls $(X \cup Y)$ in $O(n^2)$ time. Let $f(n)$ be the time for computing a maximum matching in a $n$-vertex bipartite graph.  Then the running time is $O(n^{2rk} (f(n)+n^2r))$. 

The following theorem summarizes the result of this section.

\begin{theorem}
Given a set of $n$ balls in the Euclidean plane with $k$  different types of radii where the ball centers are contained on $r$ parallel planes, a maximum clique in the corresponding ball graph can be computed in $O(n^{2rk} (f(n)+n^2r))$ time. Here $f(n)$ is the time to compute a maximum matching in a $n$-vertex bipartite graph. 
\end{theorem}

\subsection{Ball Centers are on $r$ Planes which are Perpendicular to a Different Plane}\label{sec:para2}
We now show that the condition for the input planes being parallel to each other can be relaxed as long as the input planes are all perpendicular to a different plane. Without loss of generality assume that the $r$ input planes are perpendicular to the xz-plane (e.g., Figure~\ref{fig:poly}(b)). We first observe the following property of a maximum clique. 

\begin{lemma}
    Let $B_k$ be a ball graph with $k$ different radii with centers on $r$ planes that are perpendicular to the xz-plane. Let  $\mathcal{C}_i$ be all the balls of type $i$ in a maximum clique of $B_k$. Let $P$ be the projection of the ball centers in $\mathcal{C}_i$ on the xz-plane. Then the convex-hull boundary of $P$ contains at most $2r$ vertices (i.e., they correspond to strictly convex corners).
\end{lemma}
\begin{proof} 
 If three points of $P$ come from the same original plane, then they are collinear in the xz-plane, thus only two of them can appear as strictly convex corners on the convex-hull boundary of $P$. Therefore, the $r$ planes can contribute to at most $2r$ vertices in total. 
    %
%
\end{proof}

Let $Q$ be a set of points in three dimensions. By a \emph{convex hull} of $Q$ we denote the smallest convex polyhedra that contains all points in $Q$. A \emph{lower (upper) envelope} of $Q$ consists of all points that lie at the boundary of the convex hull of $Q$ such that the rays that start at these points and move along the negative y-axis (positive y-axis) do not contain any interior point of the convex hull. An \emph{extended lower envelope} of $Q$ is a three-dimensional region that consists of the points on the lower envelope and all points that are hit by the rays that start at the lower envelope and move along the positive y-axis. Figure~\ref{fig:poly}(c) illustrates an extended lower envelope. We also define an \emph{extended upper envelope} of $Q$ symmetrically.  

Similar to Section~\ref{sec:kradii}, we guess the number of different types of radii that may appear in a maximum clique. For each ball type $i$, we guess at most $2r$ ball centers $\{a^1_i,b^1_i\},\ldots, \{a^r_i,b^r_i\}$ from $\mathcal{C}_i$, where $a^j_i$ is the leftmost  and $b^j_i$ is the rightmost over all the ball centers in $\mathcal{C}_i$ (e.g., Figure~\ref{fig:poly}(b)) on the $j$th plane. It is straightforward to observe that the convex hull of the projection of the guessed ball centers on the xz-plane contains the projection of all ball centers of $\mathcal{C}_i$. Let $\Psi$ be the set of at most $2rk$ balls that we guess over all radii types. 

For every extended lower envelope $U_i$ for the ball centers guessed for type $i$, we construct a set $\mathcal{X}_i$ by taking every ball that has its center in $U_i$ and intersects all the balls of $\Psi$. We show that the union of these balls, i.e., $\mathcal{X}= (\mathcal{X}_1\cup \mathcal{X}_2\cup\ldots\cup \mathcal{X}_i)$, determines a clique. Similarly, for each extended upper envelope $\overline{U}_i$, we construct a set $\mathcal{Y}_i$ by taking every ball that has its center in $\overline{U}_i$ and intersects all the balls of $\Psi$. We show that their union $\mathcal{Y}= (\mathcal{Y}_1\cup \mathcal{Y}_2\cup\ldots\cup \mathcal{Y}_i)$, determines a clique. It now suffices to show that the balls in $\mathcal{X}$ (similarly, the balls in $\mathcal{Y}$) are mutually adjacent.

\begin{lemma}\label{lem:polymagic}
    For every $i,j$, where $1\le i,j\le k$, the balls in $\mathcal{X}_i\cup \mathcal{X}_j$ are mutually adjacent. 
\end{lemma}
\begin{proof}
Let $B_{p,r_i}$ and $B_{q,r_j}$ be two balls, where $p$ belongs to $\mathcal{X}_i$ and $q$ belongs to $\mathcal{X}_j$. We now show that  $B_{p,r_i}$ and $B_{q,r_j}$ must mutually intersect.

The argument in the rest of the proof works irrespective of whether $i=j$ or not. Without loss of generality assume that the $y$-coordinate of $p$ is larger than or equal to that of $q$. Let $R$ be the ray that starts at $p$ and passes through $q$, and let $o$ be the intersection point of $R$ when $R$ exits the extended lower envelope $U_i$. Figure~\ref{fig:poly}(d) illustrates such a scenario when $i\not= j$. 

Assume first that $o$ hits a face $F$ of the lower envelope, and let $L$ be the plane determined by $F$. Let $m$ be the point of $L$ that minimizes the distance $|pm|$. Since $F$ is convex, moving $o$ away from $m$ on $F$ would increase the length $|po|$ monotonically and hit a vertex $w$ of $F$. Since  $B_{p,r_i}$ intersects $B_{w,r_j}$, and since $|pq|\le |po|\le |pw|\le r_i+r_j$, the balls $B_{p,r_i}$ and $B_{q,r_j}$ must intersect. 

Assume now that $o$ does not belong to the lower envelope but hits a side $F$ of the extended lower envelope, and let $L$ be the plane determined by $F$. Let $m$ be the point of $L$ that minimizes the distance $|pm|$. Since $F$ is convex, moving $o$ towards negative y-axis and away from $m$ on $F$ would increase the length $|po|$ monotonically and hit a vertex $w$ of $F$. Since  $B_{p,r_i}$ intersects $B_{w,r_j}$, and since $|pq|\le |po|\le |pw|\le r_i+r_j$, the balls $B_{p,r_i}$ and $B_{q,r_j}$ must intersect. 
\end{proof}

The following theorem summarizes the result of this section.

\begin{theorem}
Given a set of $n$ balls in the Euclidean plane with $k$  different types of radii where the ball centers are contained on $r$ planes that are perpendicular to a single different plane, a maximum clique in the corresponding ball graph can be computed in $O(n^{2rk} (f(n)+n^2r))$ time. Here $f(n)$ is the time to compute a maximum matching in a $n$-vertex bipartite graph. 
\end{theorem}

\section{Conclusion}

Given the simplicity of our algorithms, one may wonder why obtaining a positive result appeared to be challenging even for the case when we have only two different types of disks.  The existing lens-based algorithms for the unit-disk case act as a natural motivation for designing lens-based regions for the two radii case to obtain a cobiparte graph, where a maximum clique can be obtained in polynomial time.  We first present the challenges that appear in a lens-based analysis and then discuss directions for future research.

\subsection{The Deception of the Lens}

We refer the reader to Figure~\ref{fig:lbc} for an example of a lens-based attempt where the radii of the small and big disks are denoted by $r_s$ and $r_b$, respectively.  The two blue disks $D_{p,r_s}$ and $D_{p,r_s}$ in Figure~\ref{fig:lbc}(a) are guessed to be the farthest pair of small disks in a maximum clique. The green lens $L_b$ bounds how far the disk centers of the big disks can be considering that they have to intersect both $D_{p,r_s}$ and $D_{q,r_s}$. The line $\ell_{pq}$ does split the blue lens $L_s$ such that the small disks in the upper- and lower-half of $L_s$ are mutually adjacent, but the same does not hold for larger disks in $L_b$.

One can take this a step further by guessing more disks in the clique, for example, $D_{t,r_b}$ is a big disk that maximizes the distance $\max\{|pt|,|qt|\}$. Figure~\ref{fig:lbc}(b) shows that the upper half of $L_s$ could be refined further based on $D_{t,r_b}$, i.e., the red circular arc bounds how far the small disks can be as they should intersect $D_{t,r_b}$. The corresponding region is shaded in gray and the set of corresponding small disks are named slice-A. The bottom half of $L_s$ contains the rest of the small disks and this set is named slice-C. 

Consider now the line $\ell_{qt}$ Figure~\ref{fig:lbc}(c), which is a new candidate that one can examine to split big disks inside $L_b$. One can further refine $L_b$ based on where the big disks can be considering they must intersect  $D_{t,r_b}$ and $D_{q,r_s}$. The split regions determined by $\ell_{qt}$ are shaded in gray and green in  Figure~\ref{fig:lbc}(d), and the corresponding sets of big disks are named slice-B and slice-D, respectively. One can prove (with some non-trivial observations) that the disks in each individual slice are mutually adjacent, and that the disks in slice-A are adjacent to the disks in slice-B. However, to obtain a cobipartite graph we also need all disks in slice-C to be adjacent to all disks in slice-D. However, we can construct examples where a small disk in slice-C may not intersect a big disk in slice-D. One can keep adding additional guesses for disks in the maximum clique but it is challenging to formulate a cobipartite graph unless we have a right split.

\begin{figure}[pt]
    \centering
    \includegraphics[width=\textwidth]{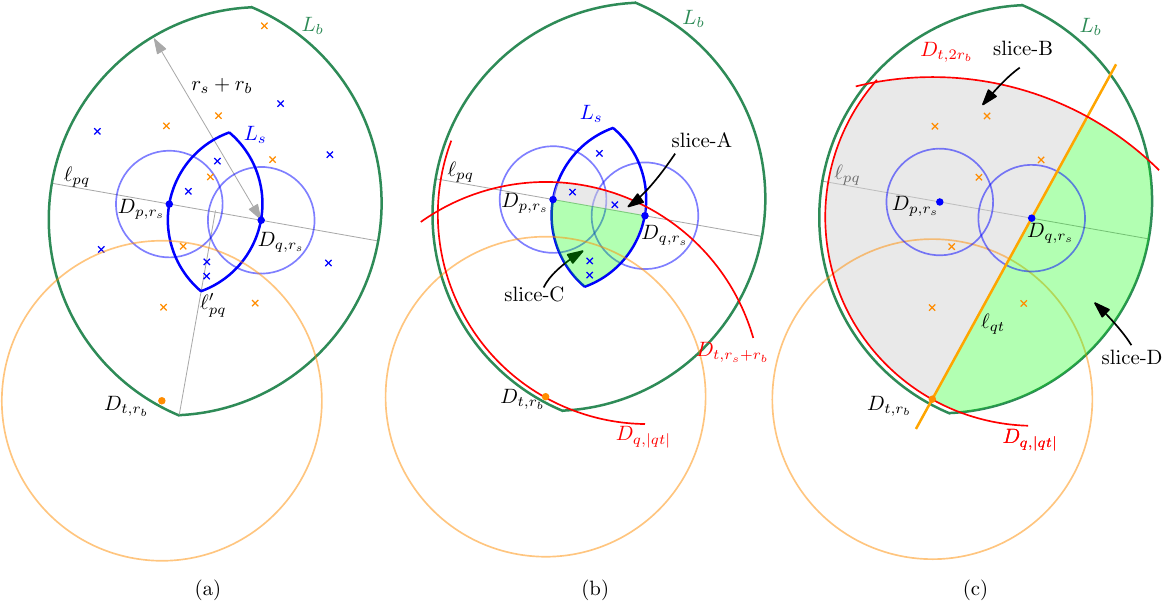}
    \caption{Lens-based regions designed for the case of two radii types.}
    \label{fig:lbc}
\end{figure}

\subsection{Direction for Future Research}
In this paper we gave an $O(n^{2k} poly(n))$-time algorithm to find a maximum clique in a $n$-vertex disk graph with $k$ different radii, and an $O(n^{2rk} poly(n,r))$-time algorithm  to find a maximum clique in a $n$-vertex ball graph with $k$ different radii where the ball centers are contained in $r$ planes that are perpendicular to a single different plane. Designing faster algorithms would be a natural direction to explore. One may   attempt to remove the perpendicularity  constraint for the case of ball graphs. The lower bound on the time complexity is another avenue to explore. The NP-hardness reduction for computing a maximum clique in a ball graph~\cite{DBLP:conf/focs/BonamyBBCT18} allows for arbitrary radii whereas our result provides polynomial-time algorithms when $r\in O(1)$ and $k\in O(1)$. Therefore, it would be interesting to examine cases when $k\in o(n)$.  

We also showed that given a set of $n$ unit disks, one can compute a maximum clique for each possible axis-aligned rectangle determined by the input disk centers in $O(n^5\log n)$ time. Consequently, given a rectangular range query $R$, the maximum clique $C$ inside $R$ can be reported in $O(\log \log n+|C|)$ time, where the $O(\log \log n)$ term is to locate the rectangle $R'$ (using range searching data structures~\cite{DBLP:conf/soda/Nekrich21}) for which a solution is precomputed. Improving the $O(n^5\log n)$ running time or establishing tight time-space trade-offs can be interesting. 

One can also examine time-space trade-offs of rectangle or disk queries of bounded size, even in the unit disk case. For example, if a query disk is of radius 1 unit, then the problem of computing a clique in the query region becomes equivalent to reporting the number of disk centers inside the query disk~\cite{rosen1999handbook}. Supporting disk queries with disks of radius $(1+\epsilon)$, where $\epsilon>0$, becomes more challenging. In fact, no polynomial-time algorithm is known for finding a maximum clique in a disk graph of arbitrary radii; even when the radii are in the interval $[1,1+\epsilon]$ and the disk centers are within a disk of radius $(1+\epsilon)$, where $\epsilon>0$. 

\section*{Acknowledgements}
The work is supported in part by the Natural Sciences and Engineering Research Council of Canada (NSERC). We thank Soichiro Yamazaki for useful feedback.

\bibliographystyle{abbrv}
\bibliography{ref}

\end{document}